\documentclass[letterpaper, 10 pt, conference]{ieeeconf}

\IEEEoverridecommandlockouts                         
\title{\LARGE \bf
Data-Driven Robust Backward Reachable Sets for Set-Theoretic Model Predictive Control 
}

\author{Mehran Attar and Walter Lucia% <-this % stops a space
	\thanks{This work was supported in part by the Natural Sciences and Engineering Research Council of Canada (NSERC).}%
	\thanks{Mehran Attar and Walter Lucia are with the Concordia Institute for Information Systems Engineering (CIISE), Concordia University, Montreal, QC, H3G 1M8, CANADA, {\tt\small mehran.attar@concordia.ca}, {\tt\small walter.lucia@concordia.ca}}%	
}

\usepackage{cite}
\let\proof\relax

\usepackage{amsthm}
\usepackage{amsmath}
\usepackage{amsfonts}
\usepackage{mathtools}
\usepackage{amssymb}
\usepackage{graphicx}
\usepackage{textcomp}
\graphicspath{{\images}}
\usepackage{siunitx}
\usepackage{ragged2e}
\usepackage{todonotes}
\usepackage{float}
\usepackage{lipsum}
\theoremstyle{plain}
\usepackage{xhfill}
\usepackage{booktabs}
\usepackage{hyperref}
\usepackage{courier}
\usepackage{algpseudocode}
\usepackage{algorithm}

\newcommand{\rr}{\mathop{{\rm I}\mskip-4.0mu{\rm R}}\nolimits}

\theoremstyle{definition}

\newtheorem{lemma}{Lemma}
\newtheorem{proposition}{Proposition}

\theoremstyle{remark}
\newtheorem{remark}{\textbf{Remark}}
\theoremstyle{remark}
\newtheorem{assumption}{Assumption}
\newtheorem{definition}{Definition}

\newtheorem{Property}{\textbf{Property}}

\begin{document}

\maketitle
\thispagestyle{empty}
\pagestyle{empty}
%%%%%%%%%%%%%%%%%%%%%%%%%%%%%%%%%%%%%%%%%%%%%%%%%%%%%%%%%%%%%%%%%%%%%%%%%%%%%%%%
\begin{abstract}
In this paper, we propose a novel approach for computing robust backward reachable sets from noisy data for unknown constrained linear systems subject to bounded disturbances. In particular, we develop an algorithm for obtaining zonotopic inner approximations that can be used for control purposes. It is shown that such sets, if built on an extended space including states and inputs, can be used to embed the system's one-step evolution in the computed extended regions. Such a result is then exploited to build a set-theoretic model predictive controller that, offline, builds a recursive family of robust data-driven reachable sets and, online, computes recursively admissible control actions without explicitly resorting to either a model of the system or the available data. The validity of the proposed data-driven solution is verified by means of a numerical simulation and its performance is contrasted with the model-based counterpart.
\end{abstract}
%%%%%%%%%%%%%%%%%%%%%%%%%%%%%%%%%%%%%%%%%%
\section{Introduction}\label{sec:introduction}

Forward reachable sets (FRS) and backward reachable sets (BRS) are important set-theoretic concepts in control theory because they allow us to analyze and predict admissible behaviors of a system over time. Such concepts are particularly important for constrained safety-critical systems where it is imperative to ensure that given unsafe configurations 
(e.g., vehicle collisions) 
are not reached \cite{mitchell2007comparing,althoff2010reachability}.

In the literature, FRS and BRS have been largely employed to design robust and model predictive control strategies, see, e.g., \cite{bertsekas1972infinite,blanchini2008set,berberich2020data,alanwar2022robust,gupta2023computation,angeli2008ellipsoidal} and references therein.
FRS and BRS are typically represented by means of polytopes, ellipsoidal sets, or zonotopes, and each representation has its own benefits/drawbacks in terms of accuracy and computational complexity~\cite{blanchini2008set, angeli2008ellipsoidal, althoff2010reachability}. Polytopes are closed under Minkowski set sum/difference; consequently, for 
linear systems subject to polyhedral constraints and disturbances, polytopes allow the exact computation of FRSs and BRSs. However, if $N$-step FRS and BRS are of interest, then such a representation suffers from an increasing number of vertices that limits its use for low-dimensional systems~\cite{yang2021scalable}. Consequently, approximated FRS and BRS representations via ellipsoids or zonotopes are often used. The use of ellipsoidal shapes is very compact, but the resulting reachable sets might be conservative \cite{kurzhanski1997ellipsoidal}; on the other hand, zonotopes are interesting because they are closed under Minkowski sums and linear mappings, and they typically allow to achieve a good compromise between accuracy and compactness of representation \cite{scott2016constrained,althoff2015introduction,girard2005reachability}.
FRS and BRS can be exactly or approximately computed by resorting to a mathematical description of the plant's dynamics or resorting to a given collection of input-output data. Model-based reachability relies on a priori accurate knowledge of the system's dynamics and allows accurate predictions; in contrast, data-driven reachability is inferred from a set of observed system's trajectories, and it is used when the mathematical model is difficult to obtain or when the system's behavior is affected by unknown or uncertain factors such as disturbances \cite{alanwar2021data}.

In the literature, particular attention has been given to the problem of computing outer approximations of data-driven FRS for control purposes \cite{devonport2020data,alanwar2021data,alanwar2022robust}.
%and their use for safety verification and design of data-driven robust and model predictive controllers \cite{alanwar2022robust}. 
%
In \cite{devonport2020data},  the authors introduced two data-driven approaches for computing FRS with probabilistic guarantees. The first represents the reachability problem as a binary classification problem; the second uses a Monte Carlo sampling approach. In \cite{alanwar2021data}, starting from a collection of noisy input-state trajectories for an unknown linear system, a procedure to compute zonotopic outer approximations of the FRS is developed. Such a solution leverages the system descriptions proposed in \cite{de2019formulas} and \cite{willems2005note}, and it can be extended to compute data-driven FRS for polynomial and Lipschitz nonlinear systems. In  \cite{alanwar2022robust}, the FRS outer approximation developed in \cite{alanwar2021data} is leveraged to design a data-driven model predictive controller.
\subsection{Contribution}\label{sec:contribution}

To the best of the author's knowledge, the problem of computing robust BRS (also known as robust one-step controllable sets) from a collection of noisy input-state trajectories has not yet been explored in the literature. In this regard, this paper extends the data-driven approach developed in \cite{alanwar2021data} (computing over approximations of forward reachable sets) to compute data-driven zonotopic inner approximations of robust backward reachable sets for unknown linear systems subject to bounded state and input constraints as well as disturbances. We develop a novel augmented description of the robust one-step controllable sets, which is then used to obtain a novel data-driven implementation of the set-theoretic MPC controller developed in \cite{angeli2008ellipsoidal} (see Section \ref{sec:model_based_set_theoretic}).
A Matlab implementation of the here-developed algorithm is available at the following web-link: 
%\href{https://github.com/PreCyseGroup/Data-Driven-ST-MPC}~.
\href{https://github.com/PreCyseGroup/Data-Driven-ST-MPC}{https://github.com/PreCyseGroup/Data-Driven-ST-MPC}.

%\url{https://github.com/PreCyseGroup/Data-Driven-ST-MPC}.

%
\subsection{Notation}

Given a matrix $M\in \rr^{n\times m},$ we denote its column vectors as $m^{(j)},$ $1\leq j\leq m,$ and the right pseudo inverse as $M^\dagger$.
Given a matrix  $M\in \rr^{n \times n-1}$, its $n-$dimensional cross-product is $CP_n(M):=[\text{det}(M^{[1]}), \cdots, (-1)^{j+1}\text{det}(M^{[j]}), \cdots, (-1)^{n+1}\text{det}(M^{[n]})]^T$ \cite{mortari1997n}, and $M^{[j]} \in \rr^{n-1 \times n-1}$ is the sub-matrix of $M$ where the $j-{th}$ row is removed. 
Moreover, $0_n \in \rr^n$ and $I_n\in \rr^{n\times n}$ indicate a vector of zeros and the identity matrix, respectively.
% %
\section{preliminaries and problem formulation}\label{section:preliminaries_and_problem_formulation}

In this section, first, some preliminary definitions are provided. Then, the problem of interest is stated.

The following definitions are adapted from \cite{blanchini2008set,althoff2010computing, alanwar2022robust}:
\definition \label{def:mikowski_sum_and_difference}
Give two sets $\mathcal{S}_1 \subset \rr^s$ and $\mathcal{S}_2 \subset \rr^s$, the Minkowski/Pontryagin set sum (denoted as $\oplus$) and difference (denoted as $\ominus$) between $\mathcal{S}_1$ and $\mathcal{S}_2$ are defined as: 
\begin{equation}\label{eq:def_for_mink_sum_and_diff}
\begin{array}{rcl}
\mathcal{S}_1 \oplus \mathcal{S}_2 &=& \{s_1 + s_2: s_1 \in \mathcal{S}_1 , s_2 \in \mathcal{S}_2 \}\\
\mathcal{S}_1 \ominus \mathcal{S}_2 &=& \{s_1 \in \rr^s :s_1 + s_2 \in \mathcal{S}_1, \forall s_2 \in \mathcal{S}_2 \}
\end{array}
\end{equation}
\definition \label{def:polytope}
%(\textbf{Polytope}) 
Given $q$ halfspaces, a  polytope $\mathcal{P}$ is defined as
\begin{equation}\label{eq:polytope_def}
    \mathcal{P} = \left\{ x \in \rr^n | Cx \leq d, C\in \rr^{q \times n},d\in \rr^{q \times 1}   \right \}   
\end{equation}
\definition \label{def:zonotope} %(\textbf{Zonotope}) 
Given a center vector $c \in \rr^n$ and $p \in \mathbb{N}$ generator vectors $g^{(i)}\in \rr^n$  collected in a matrix 
$G= \left[g^{(1)}\, \ldots,\, g^{(p)}\right]  \in \rr^{n \times p},$ referred to as the  generator matrix. Then, a zonotope is defined as (using the $\mathcal{G}-$representation)
\begin{equation}\label{eq:zonotope_def}
   \!\! \mathcal{Z}(c, G)\! =\! \! \left\{\! x\in\rr^n\!: x\! =\! c\! +\! \sum_{i=1}^{p} \beta^{(i)}g^{(i)}, -1\!\leq\! \beta^{(i)}\! \leq \! 1\! \right\}
\end{equation}
\lemma \it \label{def:half_space_zonotope}
\cite[Theorem 2]{althoff2015computing} Consider $\mathcal{Z}(c,G),$ with $p$ independent generators,  $c\in \rr^n,G\in \rr^{n\times p}$. Let $v={p\choose n-1}$ 
be the number
of combinations of $n-1$ distinct generators, 
and $\mathcal{I}_i=\{\delta_1^i,\ldots,\delta_{n-1}^i\}$ the column indices associated to each $i-th$ combination.
%$\forall\,i\in \{1,\ldots,v\}$.
The polytopic $\mathcal{H}-$representation of $\mathcal{Z}(c,G)$ is
$\mathcal{Z}(c,G)=\{x\in \rr^n: Cx\leq d\},$ where 
\begin{equation}\label{eq:half_space_representation_zonotope}
\begin{array}{rcl}
    C&=&\begin{bmatrix}
       \Bar{C}\\ -\Bar{C}
    \end{bmatrix}, \Bar{C}=\begin{bmatrix}
       \Bar{C}_1\\ \vdots \\ \Bar{C}_{v}
    \end{bmatrix}, \Bar{C}_i = \frac{CP_n(G^{<\mathcal{I}_i>})^T}{\|CP_n(G^{<\mathcal{I}_i}>)\|_2},\\
    d &=& \begin{bmatrix}
       \Bar{C}c + \Delta d\\-\Bar{C}c + \Delta d
    \end{bmatrix}, \displaystyle \Delta d = \sum_{v=1}^{p}|\Bar{C}g^{(v)}|.
\end{array}
\end{equation}
and $G^{<\mathcal{I}_i>}$ contains the column in $G$ specified by $\mathcal{I}_i.$
\definition \label{def:matrix_zonotope} %(\textbf{Matrix Zonotope}) 
Given a center matrix $C \in \rr^{n \times p}$ and $q\in \mathbb{N}$ generator matrices $G_M^{(i)}\in \rr^{n\times p}$  collected in a matrix  $G_{M} = \left[G^{(1)}_{M},\, \ldots,\, G^{(q)}_{M}\right]
\in \rr^{n \times (pq)}$. Then, a matrix zonotope is 
\begin{equation}\label{eq:matrix_zonotope}
	\begin{array}{rc}
		 \mathcal{M}(C, G_{M}) = &\displaystyle \!\!\!\!\! \{X \in \rr^{n \times p}:X\! =\! C + \sum_{i=1}^{q} \beta^{(i)}G^{(i)}_{M}, \vspace{-0.1cm}\\
		 & -1\leq \beta^{(i)} \leq 1 \}
	\end{array}   
\end{equation}

\definition \label{def:matrix_polytope} 
Consider a set of $n_v>0$ vertex matrices $\mathcal{V}_P=\{V_P^{(i)}\}_{i=1}^{n_v},$ $V_P^{(i)}\in \rr^{n\times p}.$ A matrix polytope 
$\displaystyle\mathcal{M}_P(\mathcal{V}_P)\!=\!
  \{M\in \rr^{n\times p}: M\!=\!\sum_{i=1}^{n_v}\rho_iV_P^{(i)},
  \,0\leq \rho_i\leq 1,\,  \sum_{i=1}^{n_v}\rho_i=1\}.$
\begin{remark}\label{remark:convex_hull}
Note that any zonotope (matrix zonotope) is also a polytope (matrix polytope) \cite[Sec. 3.3.3]{althoff2010reachability}. Consequently, a zonotope (matrix zonotope) can always be represented as the convex hull of its vertices (matrix vertices) \cite{althoff2015introduction}. 
\end{remark}
%
%%%%%%%%%%%%%%%%%%%%%%%%%%%%%%%%%%%%%%%
\subsection{Constrained Plant Model}\label{sec:plant_model}

Consider the class of discrete-time linear time-invariant systems described by
\begin{equation}\label{eq:linear_system}
    x_{k+1} = Ax_k + Bu_k + w_k
\end{equation}
where $k\in \mathbb{Z}_+ = \{0, 1, ...\}$ and $A\in \rr^{n\times n}, B\in \rr^{n\times m}$ are the system matrices. Moreover, $x_k\in \rr^n,$  $u_k \in \rr^m,$ and $w_k \in \rr^n$ are the state, control, and disturbance vectors, respectively. In addition, it is prescribed that:
\begin{equation}\label{eq:constraints}
\begin{array}{c}
x_k \in \mathcal{X}\subset \rr^n,\quad u_k \in \mathcal{U}\subset \rr^m,\quad 
 w_k \in \mathcal{W} \subset \rr^n
\end{array}
\end{equation}
where $\mathcal{X},\,\mathcal{U},\,\mathcal{W}$ are convex, compact, and contain the origin.
\definition \label{def:RCI_set}
A set $\mathcal{T}^0 \subseteq \mathcal{X}$ is called Robust Control Invariant (RCI)  for \eqref{eq:linear_system}-\eqref{eq:constraints} if $\forall x \in \mathcal{T}^0, \exists u \in \mathcal{U}: Ax + Bu + w \in \mathcal{T}^0, \forall w \in \mathcal{W}.$
\definition \label{def:model_based_controllable_sets}
%(\textbf{Robust one-step controllable set}) 
Consider \eqref{eq:linear_system}-\eqref{eq:constraints} and a target set $\mathcal{T}^{j-1} \subseteq \mathcal{X}.$ The set of states Robust One-Step Controllable (ROSC) to $\mathcal{T}^{j-1}$ is
\begin{equation}\label{eq:ROSC-set}
\mathcal{T}^{j} = \{x \in \mathcal{X}: \exists u \in \mathcal{U}:\! x^+ \in \mathcal{T}^{j-1}, \forall w \in \! \mathcal{W} \}
\end{equation}
and $x^+:=Ax + Bu + w.$
%

%%%%%%%%%%%%%%%%%%%%%%%%%%%%%%%%%%
\subsection{Model-based Set-Theoretic MPC (ST-MPC)} 
\label{sec:model_based_set_theoretic}

In this subsection, the dual-mode MPC controller developed in \cite{angeli2008ellipsoidal}, hereafter denoted as ST-MPC, is summarized.
Given the constrained plant model \eqref{eq:linear_system}-\eqref{eq:constraints}, a stabilizing receding-horizon controller capable of fulfilling all the prescribed constraints
and capable of robustly confining, in a finite number of steps, the state trajectory into terminal RCI set $\mathcal{T}_0$ 
can be designed as follows \cite[Section II.A]{lucia2022supervisor}:\newline
\noindent - \textit{Offline:}
	\begin{enumerate}
		\item Consider the constraint-free and disturbance-free model \eqref{eq:linear_system} and compute a stabilizing state-feedback control law, e.g., $u_k = -Kx_k,$ where $(A-BK)$ is a stable matrix. Then, compute the smallest RCI region associate to $\mathcal{T}^0$ for \eqref{eq:linear_system}-\eqref{eq:constraints}, see, e.g., \cite{rakovic2006reachability}.
		\item Starting from $\mathcal{T}^0$,  recursively compute a sequence of $N>0$ ROSC sets $\left\{\mathcal{T}^j\right\}^N_{j=1},$ where
		\begin{equation}\label{eq:STMPC_control_regions}
		\mathcal{T}^j\! =\! \{x \in \mathcal{X}\!:\! \exists u \in \mathcal{U}\!: x^+\! \in \mathcal{T}^{j-1}, \forall w \in \mathcal{W} \} 
	\end{equation}
	\end{enumerate}
\noindent - \textit{Online ($\forall\,k$):}
	\begin{enumerate}
		\item Find  $j_k := \displaystyle \min_{j\in \{0,\ldots,N\}} \{j: x_k \in \mathcal{T}^j \}$
		\item \textbf{If} $j_k = 0,$ \textbf{then} $u_k = - Kx_k.$ \textbf{Else} solve the following Quadratic Programming (QP) problem:
		\begin{equation}\label{eq:opt_for_computed_control_commands}
			\begin{array}{c}
				u_k = \arg\min\limits_{u\in \mathcal{U}} J(x_k, u) \text{ s.t. } \\
				Ax + Bu \in (\mathcal{T}^{j_k - 1}\ominus \mathcal{W})
			\end{array}
		\end{equation}
	where $J(x_k, u)$ is a convex cost function.
	\end{enumerate}

\begin{Property} \label{property:properties-ST_MPC} \it 
The ST-MPC algorithm enjoys the following properties \cite{angeli2008ellipsoidal}:
\begin{enumerate}
	\item The optimization \eqref{eq:opt_for_computed_control_commands} enjoys recursive feasibility.
	\item The terminal region $\mathcal{T}^0$ is reached in at most $j_0$ steps (with $j_0$ the set membership index at $k=0$) regardless of any admissible disturbance realization.
  %and for any admissible $J(x_k, u)$.
\end{enumerate}
\end{Property}
\subsection{Problem Statement}\label{sec:problem_statement}
\begin{assumption} \label{assumption_rank} \it
The matrices $A,$ $B$ in \eqref{eq:linear_system} are unknown, and the disturbance set $\mathcal{W}$ can be represented (or over-approximated) as a zonotope described in the $\mathcal{G}-$ or $\mathcal{H}-$ representation, i.e.,
\begin{equation}\label{eq:disturbance_zonotope}
\begin{array}{rcl}
 \mathcal{W}&=&\mathcal{Z}_w(c_w,G_w)  
 =
 %\\ &=& 
 \{w\in \rr^n: H_w w\leq h_w\},\,
 \end{array}
\end{equation}
where $c_w\in \rr^n,\,G_w\in \rr^{n\times p_w}$ and $H_w\in \rr^{n_w\times n}, h_w\in \rr^{n_w}.$
On the other hand, the state and input constraint sets can be either zonotopes or polytopes, and they are described using the $\mathcal{H}-$representation
\begin{equation}\label{eq:constraints_h_representations}
\mathcal{X} \!=\! \left\{\!x\! \in\! \rr^n\!:\! H_xx \leq h_x  \!\right\},\, \mathcal{U} \!=\! \{ \!u\!\in\! \rr^m\!:\! H_uu\leq h_u \!\}  
\end{equation}
where $H_x\in \rr^{n_x \times n},H_u\in \rr^{n_u \times m},h_x\in \rr^{n_x}, h_u\in \rr^{n_u}.$
Moreover, the following collection of $N_t>0$ input-state trajectories is available:
\begin{equation}\label{eq:available_trajectories}
	\left\{\left\{u^{(i)}_k\right\}^{N_{s}^{(i)}-1}_{k=0}\!\!\!\!,\,  \left\{x^{(i)}_k\right\}^{N_{s}^{(i)}-1}_{k=0}\right\}_{i=1}^{N_t}
\end{equation}
where $N_{s}^{(i)}>0$ is the number of samples in each trajectory. 
Moreover, the matrix $\begin{bmatrix}
X_-^T & U_-^T
\end{bmatrix}^T$ has full row rank, i.e., 
\begin{eqnarray}
&& \text{rank}(\begin{bmatrix}
X_-^T & U_-^T
\end{bmatrix}^T)=n+m \label{eq:rank_condition}\\
\!\!\!X_{-}\!\!\!&\!\!\!=\!\!&\!\!\! \left[
    x^{(1)}_0, \cdots, x^{(1)}_{N_{s}^{(1)}-1}, \cdots, x^{(N_t)}_0, \cdots, x^{(N_t)}_{N_{s}^{(N_t)}-1}\right] \\
\!\!\!U_{-}\!\!\!&\!\!\!=\!\!&\!\!\! \left [
    u^{(1)}_0, \cdots, u^{(1)}_{N_{s}^{(1)} -1}, \cdots, u^{(N_t)}_0, \cdots, u^{(N_t)}_{N_{s}^{(N_t)} -1}
    \right ] 
\end{eqnarray}
and $X\in \rr^{n \times (N_s+1)N_t},$ $X_{-}\in \rr^{n \times N_sN_t},$ 
%$X_-\in \rr^{n \times N_sN_t}$ 
and  $U_{-}\in \rr^{m \times N_sN_t}.$
\end{assumption}

\remark The condition~\eqref{eq:rank_condition} is fairly standard in the related literature, see, e.g.,\cite{de2019formulas} and references therein. In particular, it ensures that the collected data \eqref{eq:available_trajectories} have been obtained for sufficiently persistent exciting input sequences and that
\begin{equation}\label{eq:compute_AB_without_noise}
\begin{array}{c}
\begin{bmatrix}
    A & B
    \end{bmatrix} =  X_+ \begin{bmatrix}
    X_- \\ U_-
    \end{bmatrix}^{\dagger}\\
    X_{+} =\left[
    x^{(1)}_1, \cdots, x^{(1)}_{N_{s}^{(1)}}, \cdots, x^{(N_t)}_1, \cdots, x^{(N_t)}_{N_{s}^{(N_t)}}\right ]
\end{array}
\end{equation}
Moreover, the set descriptions \eqref{eq:disturbance_zonotope}-\eqref{eq:constraints_h_representations} are similar to the ones made in \cite{alanwar2021data}, with the only difference being that we do not restrict the constraints \eqref{eq:constraints} to be described by zonotopes.
\noindent \textbf{Problem of interest:} { \it
Given the input-state trajectories \eqref{eq:available_trajectories} collected for the linear model \eqref{eq:linear_system}-\eqref{eq:constraints} with unknown system matrices $(A,B):$
\begin{enumerate}
	\item Design a data-driven algorithm computing an inner approximation of the ROSC set \eqref{eq:ROSC-set}.
	\item Design a Data-Driven Set-Theoretic MPC (D-ST-MPC) controller for \eqref{eq:linear_system}-\eqref{eq:constraints} enjoying the same properties of ST-MPC (see Property~\ref{property:properties-ST_MPC}). 
\end{enumerate}
}
\section{Proposed Solution}\label{proposed_solution}

The proposed solution is developed as follows. First, we characterize the set of all system matrices $\mathcal{M}_{\hat{A}\hat{B}}$ consistent with the assumed disturbance bound $\mathcal{Z}_w$ and data \eqref{eq:available_trajectories} (see Section~\ref{sec:data_drive_representation_of_linear_systems}). Then, we show how a data-driven zonotopic inner approximation of the ROSC set \eqref{eq:STMPC_control_regions}  can be obtained from the vertex representation of the set of all compatible system matrices (see Section~\ref{sec:data_drive_inner_approx_BRS}). Finally, we design a data-driven version of the ST-MPC control strategy by leveraging a family of ROSC sets computed into an extended space domain (see Sections \ref{sec:augmented_ST_MPC} and \ref{sec:data-driven-st-mpc}).
%%%%%%%%%%%%%%%%%%%%%%%%%%%%%%%%%%%%%%%%%%%%%%%%%%%%%%%%
\subsection{Data-Driven Representation of Linear Systems With a Bounded Disturbance}\label{sec:data_drive_representation_of_linear_systems}

In the presence of noise, there generally exist multiple matrices $[A,\,  B]$ that are consistent with the data. The following Lemma describes all the linear models consistent with the available data and noise bound $\mathcal{Z}_w.$ 
\begin{lemma} \it (\cite{koch2021provably, alanwar2021data,alanwar2022robust}\label{lem:noise_zonotope}) 
Let $T =\displaystyle \sum_{i=1}^{N_t} N^{(i)}_s$ and consider the following concatenation of multiple noise zonotopes 
$$\mathcal{M}_w=\mathcal{M}_w(C_w, [G^{(1)}_{{M}_w}, \ldots, G^{(qT)}_{{M}_w}]),
$$ 
where
$C_w\in \rr^{n\times (n+m)}=[c_{w},\ldots,\,c_w]$, and 
$G_{M_w}\in \rr^{n\times T(n+m)}$ is built $\forall\,i \in \{ 1, \ldots, q\},\, \forall\, j \in \{2, \ldots, T-1\}$ as 
\begin{equation}
    \begin{array}{rcl}
        G^{(1+(i-1)T)}_{{M}_w} &=& \begin{bmatrix}
                                 g^{(i)}_{w} & 0_{n \times (T-1)}
                                \end{bmatrix}
        \\
        G^{(j+(i-1)T)}_{{M}_w} &=& \begin{bmatrix}
                                0_{n \times (j-1)} & g^{(i)}_{w} & 0_{n \times (T-j)}
                                \end{bmatrix}
        \\
        G^{(T+(i-1)T)}_{{M}_w} &=& \begin{bmatrix}
                                0_{n \times (T-1)} & g^{(i)}_{w}
                                \end{bmatrix}
    \end{array}
\end{equation}
%
%and $i =\{ 1, \cdots, p\}, j=\{2, \cdots, T-1\}.$ 
Then, the matrix zonotope
\begin{equation}\label{eq:compute_Mzono_AB}
\begin{array}{rcl}
   
\mathcal{M}_{{A} {B}}\!\!\!\!\!\!&=&\!\!\!\!\!\! (X_+ - \mathcal{M}_w) \begin{bmatrix}
    X_- \\ U_-
    \end{bmatrix}^{\dagger}   
    \\
     \!\!\!& := &\!\!\!\!\!\! \{[\hat{A},\, \hat{B}]: 
     [\hat{A},\, \hat{B}]\! =\! C_{AB} + \displaystyle\sum_{i=1}^{T} \beta^{(i)}G^{(i)}_{M_{AB}},\\
    &&
         -1\leq \beta^{(i)} \leq 1 \}
\end{array}
\end{equation}
where
$$
\begin{array}{c}
C_{AB} = (X_{+}-C_{w})\left(
    [X^T_{-},\, U^T_{-}]    ^{T}\right)^{\dagger}\\
    G_{M_{AB}} =\left[
        G^{(1)}_{M_w}
    \left(
    [X^T_{-},\, U^T_{-}]    ^{T}\right)^{\dagger}\!\!, \ldots, G^{(qT)}_{M_w}\left(
    [X^T_{-},\, U^T_{-}]    ^{T}\right)^{\dagger}
    \right]
\end{array}
$$
contains the set of all system matrices $[\hat{A},\, \hat{B}]$ that are consistent with the data \eqref{eq:available_trajectories} and disturbance bound $\mathcal{Z}_{w}$ and such that $[A,B]\in \mathcal{M}_{AB}.$ $\hfill\square$
\end{lemma}

%%%%%%%%%%%%
\subsection{Data-Driven ROSC Sets}\label{sec:data_drive_inner_approx_BRS}

If the system matrices $(A,B)$ are known the ROSC set \eqref{eq:ROSC-set} can be computed as \cite[Sec.~11.3]{borrelli2017predictive}:
\begin{equation}\label{eq:model_based_BRS}
    \mathcal{T}^{j} =  \left(\left(\mathcal{T}^{j-1} \ominus \mathcal{Z}_w) \oplus (-B  \mathcal{U} \right)\right)A \cap \mathcal{X}
\end{equation}   
In what follows, we aim at computing zonotopic inner approximations of $\mathcal{T}^j$ given the set of all consistent matrices \eqref{eq:compute_Mzono_AB}. To this end, the following issues must be addressed:
\begin{enumerate}
    \item The ROSC set $\hat{\mathcal{T}}^{j}$ consistent with $\mathcal{M}_{{A} {B}}$ and $\mathcal{Z}_{w}$ and such that
    $\hat{\mathcal{T}}^{j} \subseteq \mathcal{T}^{j}$  is:
    \begin{equation}\label{eq:ROSC_data_driven}
         \!\!\!\!\!\!\!\!\!\!\!\!\!\!\hat{\mathcal{T}}^{j}\! =\!\left\{\!\bigcap_{[\hat{A},\hat{B}]\in \mathcal{M}_{{A} {B}}} \!\!\!\!\!\!\!\!\!\left(\left(\hat{\mathcal{T}}^{j-1} \ominus \mathcal{Z}_w\right) \oplus (-\hat{B}  \mathcal{U}) \right)\!\!\hat{A}\right\}\cap \mathcal{X}
    \end{equation}
    However, since there are infinite matrices $[\hat{A},\hat{B}]\in \mathcal{M}_{{A} {B}},$ the above is not computable.
   \item According to \eqref{eq:ROSC_data_driven}, $\hat{\mathcal{T}}^{j}$ is not (in general) a zonotope even if we assume that $\hat{\mathcal{T}}^{j-1},$ $\mathcal{U}$ and $\mathcal{X}$ are all zonotopes. Indeed, zonotopes are not closed under Minkowski difference operations and set intersections.
   %required by~\eqref{eq:ROSC_data_driven}.
\end{enumerate}
The following proposition is instrumental to obtaining a computable definition of the ROSC set \eqref{eq:ROSC_data_driven} and solving the first issue.
\begin{proposition}\label{proposition:}\it Let $\mathcal{V}_{AB}$ be the set of vertices $\{\hat{A}_i,\hat{B}_i\}_{i=1}^{n_v}$ of the matrix polytope associated to $\mathcal{M}_{AB}.$ If $\hat{\mathcal{T}}^{j-1}$ is a convex set, then the ROSC set $\hat{\mathcal{T}}^{j}$ in \eqref{eq:ROSC_data_driven} 
%that are consistent with $\mathcal{M}_{AB}$ 
can be computed as
\begin{equation}\label{eq:ROSC_data_driven_vertices}
        \!\hat{\mathcal{T}}^{j}\! =\!\left\{\!\bigcap_{[\hat{A}_i,\hat{B}_i]\in \mathcal{V}_{AB}}\!\!\!\!\! \!\!\!\!\! \!\left(\left(\hat{\mathcal{T}}^{j-1} \ominus \mathcal{Z}_w\right) \oplus (-\hat{B}_i \mathcal{U}) \right)\! \hat{A}_i \right\}\cap \mathcal{X}
    \end{equation}
\end{proposition}
\proof 

Given $\mathcal{V}_{AB},$ we can write any admissible system one-step evolution as
\begin{equation}
  x^+(\alpha)=A(\alpha)x + B(\alpha)u + w  
\end{equation}
where
\begin{equation}
\begin{array}{c}
   \!\!\! \hat{A}(\alpha)= \displaystyle\sum_{i=1}^{n_v}\alpha_i A_i,\quad  
    \hat{B}(\alpha)= \displaystyle\sum_{i=1}^{n_v}\alpha_i B_i 
     \\
   \!\!\!\left[\alpha_1 , \ldots, \alpha_{n_v}\right]^T\!\!\!\in \!\mathcal{P}\! = \!\{\alpha \in \rr^{n_v}\!\!: \displaystyle \alpha_i\geq0,\,\sum_{i=1}^{n_v}\!\!\alpha_i=1\}
\end{array}    
\end{equation}
Consequently, the data-driven ROSC set \eqref{eq:ROSC-set} can be written as
\begin{equation}\label{eq:ROSC_sets_for_polytopic_model}
\begin{array}{rcl}
\displaystyle
\hat{\mathcal{T}}^j\!\!\!&\!\!=\!\!&\!\!\!\!\! \{x \in \mathcal{X}\!:\! \exists u\! \in \mathcal{U}: x^+\!(\alpha)\! \in \! \hat{\mathcal{T}}^{j-1},
         \forall w \!\in\! \mathcal{Z}_w, \forall \alpha\!\in\! \mathcal{P}\}
         \\
         \!\!\!&\!\!=\!\!&\!\!\!\!\! \{x \in \mathcal{X}\!: \exists u \in \!\mathcal{U}: \!\displaystyle\sum_{i=1}^{n_v}\!\alpha_i (\hat{A}_i x + \hat{B}_i u +w) \in \hat{\mathcal{T}}^{j-1}, \\
          && \forall w\in \mathcal{Z}_w, 
        
         \forall \alpha \in \mathcal{P}\}
\end{array}
\end{equation}
Then, definition \eqref{eq:ROSC_sets_for_polytopic_model} implies that any convex combination of the vertices must belong to $\hat{\mathcal{T}}^{j-1}.$ However, if 
$\hat{\mathcal{T}}^{j-1}$ is a convex set, then such a requirement is satisfied iff $\forall\, i=1,\ldots,n_v,$ (i.e., for any vertex model $(\hat{A}_i,\hat{B}_i)$)  $x^+_i=\hat{A}_ix + \hat{B}_iu+w\in \hat{\mathcal{T}}^{j-1},\,\,\forall w \in \mathcal{Z}_w.$ Consequently, the ROSC set \eqref{eq:ROSC_data_driven} can be computed as in \eqref{eq:ROSC_data_driven_vertices}, concluding the proof. $\hfill\blacksquare$

The second issue can be addressed directly through computation
of a polytopic set. However, in this case, the  number of vertices necessary to describe the ROSC sets will increase exponentially with the number of computed sets, making such an option not suitable for implementing the ST-MPC controller. On the other hand, a zonotopic inner approximation of the exact ROSC sets can be derived as follows. Consider the $\mathcal{H}-$ representation of $\mathcal{Z}_w$ and $\hat{\mathcal{T}}^{j-1}=\left\{\!x\! \in\! \rr^n\!:\! H_{\hat{\mathcal{T}}^{j-1}} x \leq h_{\hat{\mathcal{T}}^{j-1}}  \!\right\},
$
then, we can first compute the polytope  $\hat{\mathcal{T}}^j$ using \eqref{eq:ROSC_data_driven}  and then compute a zonotopic inner approximation of $\hat{\mathcal{T}}^j,$ using, e.g., the algorithm proposed in \cite[Section IV.A.2]{yang2021scalable} and summarized in Lemma~\ref{lemma:inner_approx}. 

\begin{remark}\label{remak:data-driven-stMPC}
The above solution to compute data-drive ROSC set $\hat{\mathcal{T}^j}$ is still not suitable to efficiently implement a data-driven ST-MPC. Indeed, if a family of zonotopic data-driven ROSC sets $\{\hat{\mathcal{T}}^j\}_{j=0}^N$ is computed as prescribed by the offline phase of the ST-MPC algorithm (see Section~\ref{sec:model_based_set_theoretic}), then the data-driven equivalent of the optimization  \eqref{eq:opt_for_computed_control_commands}
would be
\begin{equation}\label{eq:data-drive_opt_for_computed_control_commands}
			\begin{array}{c}
				u_k = \arg\min\limits_{u\in \mathcal{U}} J(x_k, u) \hspace{2mm} s.t. \\
				\hat{A}_ix + \hat{B}_iu \in (\mathcal{T}^{j_k - 1}\ominus \mathcal{Z}_w), \,\forall i=1,\ldots,n_v
			\end{array}
		\end{equation}   
Note that the above presents a number of constraints which scale exponentially with the number of matrix generators describing $\mathcal{M}_{AB}$ \cite{kochdumper2019representation}. Consequently, the resulting optimization problem presents a computational complexity much higher than the model-driven counterpart~\eqref{eq:opt_for_computed_control_commands}. In the next section, we solve such an issue by computing an augmented zonotopic representation of the ROSC sets.
\end{remark}
%%%%%%%%%%%%%%%%%%%%%%%%%%%%%%%%%%
\subsection{Augmented Data-Driven ROSC Sets}\label{sec:augmented_ST_MPC}

The model-based ROSC set \eqref{eq:ROSC-set} can be equivalently re-defined as \cite{lucia2022supervisor}:
\begin{equation}\label{eq:extended_control_regions}
	\begin{array}{rcl}
    \Xi^j\!\!\!\! &\!=\!&\!\!\!\! \left\{\!(x,u) \in \mathcal{X} \times \mathcal{U}: Ax +Bu +w\in \mathcal{T}^{j-1}, \forall w \in \! \mathcal{W}
    \right\}\\
    \mathcal{T}^{j}\!\!\!\! &\!=\!&\!\!\!\! Proj_x(\Xi^j) 
	\end{array}
\end{equation}
where $\Xi^j$ is the $(x,u)-$ augmented space description of  the ROSC set and $Proj_x(\Xi^j)$ performs a projection operation of $\Xi^j$ into the $x-$domain. Consequently, for the data-driven model described by the matrix vertices $\mathcal{V}_{AB}$, we can write 
\begin{equation}\label{eq:ROSC_data_driven_augm}
	\begin{array}{c}
	\!\hat{\Xi}^{j}_{AB} \!= \!\!\!\!\!\!\!\!\! \displaystyle  
	\bigcap_{[\hat{A}_i,\hat{B}_i]\in \mathcal{V}_{AB}}\!\!\!\!\!\!\!\!\! 
	\left\{z=[x^T,u^T]^T\! \in \rr^{n+m}\!:\!H_{z}^iz\leq h_{z}^i\right\} \vspace{0.1cm}
	\\
%	\hat{\Xi}^{j} \!=\!\hat{\Xi}^{j}_{AB}\cap \mathcal{X}, \quad 
	 \hat{\mathcal{T}}^{j}= Proj_x(\hat{\Xi}^{j}_{AB})
	\end{array}
\end{equation}
where
\begin{equation}\label{eq:H-rep_extended}
	H_z^i= \begin{bmatrix}
		H_x & 0 \\
		H_{\hat{\mathcal{T}}^{j-1}}\hat{A}_i & H_{\hat{\mathcal{T}}^{j-1}}\hat{B}_i\\
		0 & H_u
	\end{bmatrix}, 
	\quad
	h_z^i=\begin{bmatrix}
		h_x\\
		\Tilde{h}_{\hat{\mathcal{T}}^{j-1}}\\
		h_u
	\end{bmatrix} 
\end{equation}
and
\begin{equation}\label{eq:compute_tilde_h}
	\Tilde{h}_{\hat{\mathcal{T}}^{j-1}_r} = \min_{w\in \mathcal{W}}(h_{\hat{\mathcal{T}}^{j-1}_r} - H_{{\hat{\mathcal{T}}^{j-1}}_{r}} w)  
\end{equation}
\begin{remark}
In \eqref{eq:H-rep_extended}, the first and third rows describe state and input constraints, respectively, while the second row imposes, via the Minkowski difference (see \eqref{eq:compute_tilde_h}), that  $\hat{A}_ix +\hat{B}_iu +w\in \mathcal{T}^{j-1}, \forall w \in \! \mathcal{W}.$ \hfill $\Box$
\end{remark}
\begin{remark}\label{remark-expensive}
For similar reasons to the ones explained in Section~\ref{sec:data_drive_inner_approx_BRS}, the augmented ROSC set $\hat{\Xi}^j$ is a polytope. Moreover, the projection operation required to obtain $\hat{\mathcal{T}}^j$ is computationally demanding if performed on a polytope \cite{yang2021scalable}. 
\end{remark}
To obtain a zonotopic inner approximation of $\hat{\Xi}^j$ we can resort to the following lemma:
\begin{lemma} \it \label{lemma:inner_approx} (Adapted from \cite[Section IV.A.2]{yang2021scalable})
Consider a template zonotope $\mathcal{Z}(c,G),$ with $c\in \rr^{n+m}$ denoting any center vector  and $G\in \rr^{(n+m) \times p}$ a fixed set of chosen generators.  The set $\hat{\Xi}^j_z = \mathcal{Z}(c^*,\beta^* G)$ is an inner approximation of $\hat{\Xi}^{j}$ (i.e., $\hat{\Xi}^j_z\subseteq \hat{\Xi}^j$) if $c^*\in \rr^{n+m},\beta^*=\in \rr^{p\times 1}$ are computed solving the following optimization problem
\begin{equation}\label{eq:zonotopic_inner_approx}
	\begin{array}{c}
	\{c^*,\beta^*\}=
	  \arg\max\limits_{c,\,\beta=[\beta_1,\ldots,\beta_l]^T} \sum\limits_{l=1}^{p} d_l \log(\beta_l),\quad s.t. \vspace{1mm} \\
		 H_{\hat{\Xi}^{j}}c + |H_{\hat{\Xi}^{j}}G|\beta \leq h_{\hat{\Xi}^{j}}, 0\leq \beta \leq 1
	\end{array}
\end{equation}
where $(H_{\hat{\Xi}^{j}},h_{\hat{\Xi}^{j}})$ denotes the $\mathcal{H}-$representation of  $\hat{\Xi}^{j},$ $d_l \geq 0$ are weighting factors,  and $|H_{\hat{\Xi}^{j}}G|$ denotes a matrix obtained taking the element-wise absolute value of $H_{\hat{\Xi}^{j}}G$.  
\end{lemma}

By denoting with  $\text{In}_z(\hat{\Xi}^{j})$ the   approximation performed by \eqref{eq:zonotopic_inner_approx}, the inner zonotopic approximation of  \eqref{eq:ROSC_data_driven_augm} is   
\begin{equation}\label{eq:inner_ROSC_data_driven_augm}
		\begin{array}{c}
		\!\hat{\Xi}^{j}_{AB} \!= \!\!\!\!\!\!\!\!\! \displaystyle  
		\bigcap_{[\hat{A}_i,\hat{B}_i]\in \mathcal{V}_{AB}}\!\!\!\!\!\!\!\!\! 
		\left\{z=[x^T,u^T]^T\! \in \rr^{n+m}\!:\!H_{z}^iz\leq h_{z}^i\right\} \vspace{0.1cm}
		\\
		\hat{\Xi}^{j}_z \!=\text{In}_z\left\{\hat{\Xi}^{j}_{AB}\right\}, \quad 
		\hat{\mathcal{T}}^{j}_z= Proj_x(\hat{\Xi}^j_z)
	\end{array}
\end{equation}
%
%%%%%%%%%%%%%%%%%%%%%%%%%%%%%%%%%%
\subsection{Data-Driven Set-Theoretic MPC}\label{sec:data-driven-st-mpc}

It is now shown how the augmented ROSC sets \eqref{eq:inner_ROSC_data_driven_augm} can be used to implement a data-driven ST-MPC strategy, solving the issue discussed in Remark~\ref{remak:data-driven-stMPC}. 

\begin{proposition}
\it
The Data-Driven Set-Theoretic MPC controller described in Algorithm~\ref{algorithm:set_theoretic_MPC} 
enjoys the same properties as the model-based ST-MPC (see Property~\ref{property:properties-ST_MPC})
\end{proposition}

\begin{proof}
First, in the offline phase, the proposed algorithm computes, using \eqref{eq:ROSC_data_driven_vertices}, a family of pair of ROSC sets $\{\hat{\Xi}_z^j,\hat{\mathcal{T}^j_z}\}_{j=1}^N,$ where  $\hat{\mathcal{T}^j_z}$ is an inner zonotopic approximation of \eqref{eq:STMPC_control_regions} and $\hat{\Xi}_z^j$ is an augmented ROSC set representation of $\hat{\mathcal{T}^j_z}$ ensuring, by construction, that if $[x,\,u]^T\in \hat{\Xi}_z^j$ then $u\in \mathcal{U}$ and $\hat{A}_ix+\hat{B}_iu+w\in \hat{\mathcal{T}^{j-1}_z},\forall\,w\in \mathcal{W}, \forall \{\hat{A}_i,\hat{B}_i\}\in \mathcal{V}_{AB}.$ Consequently, $\hat{\Xi}_z^j$ implicitly embeds the worst-case admissible one-step evolution of the underlying linear system and the optimization \eqref{eq:control_based_extended_control_regions} defines the data-driven equivalent of \eqref{eq:opt_for_computed_control_commands}. Since by construction \eqref{eq:opt_for_computed_control_commands} is recursively feasible, then the data-driven MPC Algorithm~\ref{algorithm:set_theoretic_MPC} 
enjoys the same properties as its model-based counterpart.
\end{proof}
\begin{algorithm}[h!] \label{algorithm:set_theoretic_MPC}
\textbf{Input:}
Ver$(\mathcal{M}_{AB})$, RCI set $\mathcal{T}^0$, $\mathcal{X}$, $\mathcal{U}$ and $\mathcal{W}$\;
\vspace{0.15cm}

 \noindent
	\xrfill[0.7ex]{1pt}Offline\xrfill[0.7ex]{1pt}
\begin{algorithmic}[1]
\State Let $\hat{\mathcal{T}}^0_z = \mathcal{T}^0$ a given terminal RCI set\;
\State Compute the ROSC sets $\{\hat{\Xi}^{j}_z$, $\hat{\mathcal{T}}^{j}_z\}_{i=1}^N$ using \eqref{eq:inner_ROSC_data_driven_augm}
\end{algorithmic}
\vspace{0.15cm}
\noindent
	\xrfill[0.7ex]{1pt}Online ($\forall\,k$)\xrfill[0.7ex]{1pt} 
\begin{algorithmic}[1]
 \State Find set membership index $j_k := \displaystyle \!\!\!\!\!\! \min_{j\in \{0,\ldots,N\}} \{\!j\!: \! x_k \in \! \hat{\mathcal{T}}^j_z \!\}$\;
\If{$j_k=0$}{   $u_k = -Kx_k$}
\Else{ compute $u_k$ by solving the following QP problem
\begin{equation}\label{eq:control_based_extended_control_regions}
	\begin{array}{c}
		u_k = \arg\min J(x_k, u) \quad s.t. \\
	\left[x_k^T,u^T\right]^T \in \hat{\Xi}^j_z
	\end{array}
\end{equation}}

\EndIf
\end{algorithmic}
\caption{Data-Driven Set-Theoretic MPC (D-ST-MPC)}
\end{algorithm}
\begin{remark}\label{remark:computing_terminal_region}
	The computation of the RCI region $\mathcal{T}^0$	is outside of the scope of this paper. However, it can be computed/estimated using data-driven strategies \cite{chen2021data}. In alternative, $\mathcal{T}^0$ can be chosen as an arbitrary small set that is only used to initialize the ROSC set computation. With such a solution, the proposed strategy is feasible if $\mathcal{T}^0 \subseteq \hat{\mathcal{T}}_z^1.$ Indeed, according to Definition~\ref{def:RCI_set}, if ${\mathcal{T}}^0 \subseteq \hat{\mathcal{T}}_z^1,$ then  ${\mathcal{T}}^0$ is an RCI set and the associated terminal control law for any $x_k\in {\mathcal{T}}^0$ is
	\begin{equation}\label{eq:terminal_controller_noRCI}
		\begin{array}{c}
			u_k = \arg\min J(x_k, u) \quad s.t.  \\
			\left[x_k^T,u^T\right]^T \in \hat{\Xi}^1_z
		\end{array}
	\end{equation}
\end{remark}

\section{Illustrative Example}\label{sec:simulation}
Consider the discrete-time linear time-invariant system described \cite{rakovic2006reachability}, where  %
\begin{equation}\label{eq:simulation_example}
  A= \begin{bmatrix}
     0.7969 & -0.2247\\0.1798 & 0.9767
  \end{bmatrix},\quad B=\begin{bmatrix}
     0.1271\\0.0132
  \end{bmatrix}
\end{equation}
$|u_k| \leq 3, |x_k^l|\leq 10,\,l=1,2$   and the disturbance set is $\mathcal{Z}_w (0_2,0.005I_2)$.
By assuming that the pair $(A,B)$ is unknown, we have
simulated the system  considering random admissible inputs to collect two input-state trajectories \eqref{eq:available_trajectories}, each containing $N_s=10$ samples.  After verifying that the rank condition \eqref{eq:rank_condition} is fulfilled, we have computed  the set $\mathcal{M}_{AB}$ of all system matrices consistent with \eqref{eq:available_trajectories} and $\mathcal{W}.$ Then, we computed the vertices $\mathcal{V}_{AB}$ of $\mathcal{M}_{AB}$ and used \eqref{eq:inner_ROSC_data_driven_augm} to build a family of $N=15$ data-driven ROSC sets  $\{\hat{\Xi}^{j}_z,\,\hat{\mathcal{T}}^{j}_z\}$ (dashed red zonotopes in Fig.~\ref{fig:controllable_regions}). Since $\hat{\mathcal{T}}_z^1 \subseteq \hat{\mathcal{T}}_z^2,$  $\hat{\mathcal{T}}_z^1$ has treated as the terminal RPI set with control law \eqref{eq:terminal_controller_noRCI}, see Remark~\ref{remark:computing_terminal_region}.
To contrast the proposed data-drive ST-MPC with its model-based counterpart, we have also computed a model-based RCI set $\mathcal{T}_0$ using $K=[-0.47, -0.19]$ and the algorithm in \cite{rakovic2005invariant}. Then,  $N=5$ model-based ROSC sets $\{\mathcal{T}^j\}$ has been computed as prescribed by \eqref{eq:STMPC_control_regions} (solid black polytopes in Fig.~\ref{fig:controllable_regions}). Note that the discrepancies in the number of ROSC sets computed for the model-based and data-driven approaches are justified by the fact that the data-driven approach computes inner zonotopic approximations of the actual polytopic sets computed with the model. 
The results shown in Figs.~\ref{fig:controllable_regions}-\ref{fig:control_inputs} have been obtained for an initial plan's initial state $x_0=[-2,\, 1.1]^T$ (i.e., $x_0\in \hat{\mathcal{T}}^{15}_z$ and $x_0\in {\mathcal{T}}^{5}$) and configuring the ST-MPC and D-ST-MPC schemes to minimize the control effort, i.e.,  $J(x_k,u)=0.5u^2.$ In Fig.~\ref{fig:controllable_regions}, it is possible to appreciate how both ST-MPC and D-ST-MPC  controller are able to confine, in a finite number of steps, the state trajectory in the RCI set $\hat{\mathcal{T}}_z^1$ and $\mathcal{T}_0,$ respectively (see the zoom-in in Fig.~\ref{fig:controllable_regions}). This is also confirmed by the set-membership index shown in Fig.~\ref{fig:control_inputs} (top subplot). In addition, in Fig.~\ref{fig:control_inputs} (bottom subplot), it is possible to note that both controllers produce control input sequences fulfilling the prescribed input constraints. As per the design, D-ST-MPC shows the same behavior as ST-MPC, however, as expected, D-ST-MPC is slightly more conservative in terms of the controller's domain of attraction (i.e., the region given by the union of the ROSC sets shown in Fig.~\ref{fig:controllable_regions}) and in terms of time required to reach the terminal RCI set ($5$ steps for ST-MPC and $8$ steps for D-ST-MPC).

\begin{figure}[h!]
    \centering
    \includegraphics[width=1\linewidth]{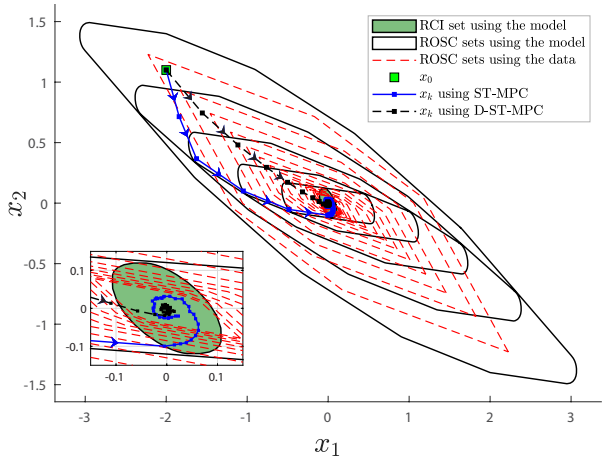}
    \caption{Model-based controllable sets $\{\mathcal{T}^j\}^5_{j=0}$ (black polyhedral) of system~\eqref{eq:simulation_example}, projection of extended space data-driven controllable sets $\{\Tilde{\Xi}^j\}_{j=0}^{15}$ on $x_1$ and $x_2$ axis (red dashed polyhedral)   
    }  
    \label{fig:controllable_regions}
\end{figure}
\begin{figure}[h!]
    \centering
    \includegraphics[width=1\linewidth]{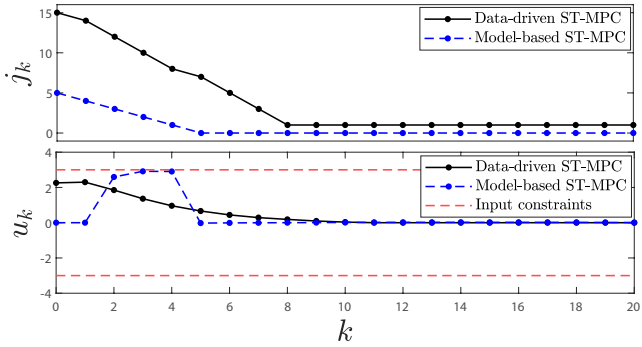}
    \caption{Set membership index (top figure) and control input signals (bottom figure)}
    \label{fig:control_inputs}
\end{figure}
%

%%%%%%%%%%%%%%%%%%%%%%%%%
\section{Conclusions}\label{conclusions}

This paper proposes a novel data-driven method for computing ROSC sets for an unknown linear system subject to bounded disturbances, state, and input constraints. In particular, it has been shown that using a collection of input-output trajectories, a zonotopic inner approximation of the backward reachable sets can be obtained by resorting to an extended space representation. The peculiar feature of the proposed solution is that the computed sets can also be used to efficiently implement a data-driven dual-mode receding horizon controller. Future studies will be devoted to mitigating the conservativeness of the proposed inner approximation for the ROSC sets. 
%%%%%%%%%%%%%%%%%%%%%%%%%

\bibliographystyle{IEEEtran}
\bibliography{bibliography} 

\end{document}